\documentclass{article}
\usepackage{amsfonts}
 \usepackage{graphicx}
\usepackage{float}

\usepackage[center,pagestyles]{titlesec}

\usepackage{paralist}
\usepackage[perpage,symbol*]{footmisc}
\usepackage{footmisc}
\usepackage{latexsym,bm,amssymb}
\usepackage{amsmath}
\usepackage{amsthm}
 \usepackage{indentfirst}
\usepackage[sort&compress,numbers]{natbib}
\usepackage[colorlinks=true]{hyperref}
\hypersetup{urlcolor=blue, citecolor=red}
\newtheorem{theorem}{Theorem}

\newtheorem{lemma}{Lemma}
\newtheorem{example}{Example}

\theoremstyle{definition}
\newtheorem{definition}{Definition}

\titleformat{\section}{\normalsize\bfseries}{\thesection}{0.5em}{}
\titlelabel{\thesection \quad}

\begin{document}

\title{A predicable condition for boundary layer separation of 2-D incompressible  fluid flows}

\author{ \small Hong Luo$^a$ \thanks{Corresponding author: lhscnu@hotmail.com} \quad Quan Wang$^b$ \quad Tian Ma$^b$ \\
{\small a. College of Mathematics and Software Science, Sichuan Normal University}\\
 {\small Chengdu, Sichuan 610066, China }\\
{\small b. Department of Mathematics, Sichuan University}\\
{\small Chengdu,
Sichuan 610021, China }\\
}
\maketitle
 {\bf Abstract}:
    In this paper, the solutions of Navier-Stokes equations with Dirichlet boundary conditions governing 2-D incompressible fluid flows are considered. A condition for boundary layer separation, which is determined by initial values and external forces, is obtained. More importantly, the condition can predict directly when and where boundary layer separation will occur. The main technical tool is geometric theory of incompressible flows developed by T. Ma and S.Wang in \cite{Ma6}


{\bf MR(2010)Subject Classification:}  35Q30, 35Q35, 76D10, 76M

\section{Introduction}
Boundary layer is the thin flows closed to the surface of the object. The concept is proposed by Prandtl in 1904. Since then, studying boundary layer becomes an important topic among mechanics, and there are many results\cite{W. E},\cite{Liu},\cite{Oleinik1},\cite{Oleinik},\cite{Schlichting}. Prandtl pointed out that boundary layer can be described by Navier-Stokes equations in \cite{Prandtl}.

 Boundary layer separation is the phenomenon that the original flows closed to the surface of object go away.
It is a very common phenomenon in geophysical dynamics, such as vortex of gulf stream, separation of atmospheric circulation near mountain and the formation of a tornado. There are may researches \cite{Blanchonette},\cite{Chorin},\cite{Ghil},\cite{Larin},\cite{Larin1},\cite{Ma5},\cite{Ma2},\cite{Ma4},\cite{Ma6},\cite{Ma8},\cite{Smith} on boundary layer separation over the past one hundred years. However, there is no mathematical theory which can show when and where boundary layer separation will occur as Chorin and Marsden pointed out in their book \cite{Chorin}.

 Our main objective is to get a condition for boundary layer separation, and the condition is determined by initial values and external forces. That is, we hope that the condition is observable. We know that boundary layer separation is essentially the structural bifurcation of the velocity field of fluid flows from\cite{Ma7},\cite{Ma5},\cite{Ma4},\cite{Ma6}. Hence, we should look for the desirable predictable condition by combining Euler dynamics and Lagrange dynamics.

       In this paper, we study boundary layer separation governed by the following Navier-Stokes equation:
 \begin{eqnarray}
\label{eql1} \left\{
   \begin{array}{ll}

u_t+ (u \cdot \nabla)u=\mu \triangle u-\frac{1}{\rho} \nabla p+f(x),
& \\
div u=0,
 \\
 u|_{\partial \Omega}=0,
& \\
 u(x,0)=\tilde{\varphi}(x),  & \\
\end{array}
\right.
\end{eqnarray}
where $\tilde{\varphi}(x)|_{\partial \Omega}=0$, $\Omega$ is bounded and open of $R^2$ with boundary $\partial \Omega$,
$\tilde{\varphi}(x)=(\tilde{\varphi}_1,\tilde{\varphi}_2)\in C^3(\Omega;R^2)$ is initial value, $f(x)=(f_1(x),f_2(x))\in C^1(\Omega;R^2)$ is
external force.

Our main result is based on analyzing the solutions of (\ref{eql1}) and the lemma of boundary layer separation developed by T. Ma, S.Wang and M.Ghil in \cite{Ma7},\cite{Ma5},\cite{Ma4},\cite{Ma6}. The result in this paper can tell us when and where boundary layer separation will occur. We can predict when a vortex will appear in the fluid flows.

 The paper is organized as follows. In section 2, we introduce preliminaries containing the concept of boundary layer separation, boundary singularity, the lemma of boundary layer separation and our main theorem. The physical interpretation related to
Theorem 1 will be given in section 3. In section 4, we will give some examples of physical application of Theorem 1.

\section{Preliminaries and the predicable condition}

Let $\Omega$ be a bounded and open domain of $R^2$. $\partial\Omega$ is $C^{r+1}$. $C^r(\Omega)$ is the
space of all $C^r$ fields on $\Omega$.
$$
B_0^r(\Omega)=\{u \in C^r(\Omega)| div u=0, u|_{\partial\Omega}=0\}.
$$
$n$ and $\tau$ is the unit normal and tangent vector of $\partial\Omega$, respectively.

We start with some basic concepts.
\begin{figure}[h]
 \begin{center}
 \includegraphics[width=9cm]{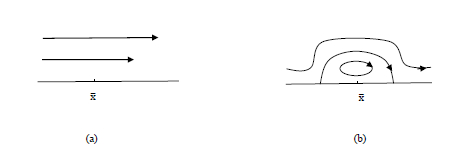} \\
 \end{center}
   \caption{}
 \end{figure}

\begin{definition}
\cite{Ma7},\cite{Ma5},\cite{Ma4},\cite{Ma6} Suppose $u\in B^r_0(\Omega)(r\geq 2)$. $\bar{x}\in \partial\Omega$ is called a boundary singularity of $u$, if $\frac{\partial u_\tau(\bar{x})}{\partial n}=0$.
\end{definition}
\begin{definition}
 \cite{Ma7},\cite{Ma5},\cite{Ma4},\cite{Ma6} We call that boundary layer separation governed by a 2-D vector field $u\in C^1([0,T];B_0^2(\Omega))$ occurs at $t_0$, if $u(x,t)$ is topologically equivalent to the structure of figure 1(a) for any $t<t_0$, but $u(x,t)$ is topologically equivalent to the structure of figure 1(b) for $t>t_0$. That is, if $t<t_0$, $u(x,t)$ is topologically equivalent to
a parallel flow,and if $t>t_0$, $u(x,t)$ separates a vortex. Furthermore, we call that boundary layer separation occurs at $\bar{x}\in \partial\Omega$, if $\bar{x}$ is an isolated boundary singularity at time $t=t_0$.
 \end{definition}
The following lemma proved by T. Ma, S. Wang and M. Ghil is necessary for our result.
\begin{lemma}
\cite{Ma7},\cite{Ma5},\cite{Ma4},\cite{Ma6} Let $u\in C^1([0,T];B_0^2(\Omega))$ be 2-D vector field and $\bar{x} \in \partial\Omega$. Boundary layer separation represented by $u$ occurs at $(\bar{x}, t_0)$,
 if there exists
$0<t_0<T$ such that
\begin{equation}\label{equ7}
\frac{\partial u_\tau}{\partial n}(x, t)\neq 0, \ \ \ t<t_0, x \in \partial\Omega
\end{equation}
 \begin{equation}\label{equ8}
\frac{\partial u_\tau}{\partial n}(\bar{x}, t)=0, t=t_0,
 \end{equation}
where $\bar{x}$ is an isolated boundary singularity of $u(\cdot, t_0)$ on $\partial\Omega$.
\end{lemma}

 Choose one party $\Gamma \subset\partial\Omega$. Without loss of generality, we take a coordinate system
$(x_1,x_2)$ with $\bar{x}$ at the origin and $\Gamma$ given by
\begin{eqnarray*}
\Gamma = \{{(x_1,0)\mid0<|x_1|<\delta}\}
\end{eqnarray*}
 for some $\delta > 0$. Obviously, the tangent and normal vectors on $\Gamma$ are the unit vectors
in the $x_1$- and $x_2$-directions, respectively.

With $\tilde{\varphi}(x)|_{\partial \Omega}=0$ and $div \varphi=0$, we get
 \begin{equation}\label{equ3}
 \tilde{\varphi}_1=x_2 \varphi_{11}(x_1)+x_2^2 \varphi_{12}(x_1)+x_2^3 \varphi_{13}(x_1)+o(x_2^3),
 \end{equation}
 \begin{equation}\label{equ4}
\tilde{\varphi}_2=x_2^2 \varphi_{21}(x_1)+o(x_2^2).
 \end{equation}

 We consider the Taylor expansions of $f(x)=(f_1(x),f_2(x))\in C^1(\Omega;R^2)$ with respect to $x_2$
 \begin{equation}\label{equ5}
 f_1=g_0(x_1)+x_2 g_1(x_1)+o(x_2),
 \end{equation}
 \begin{equation}\label{equ6}
f_2=h_0(x_1)+x_2 h_2(x_1)+o(x_2).
 \end{equation}
 Our main result is as followed.
\begin{theorem}
Suppose $\varphi \in C^3(\Omega,R^2)$ and $f\in C^1(\Omega,R^2)$, if
\begin{equation}\label{6}
0< \min_{\Gamma} \frac{-\varphi_{11}}{\mu \varphi_{11}^{''}+6 \mu\varphi_{13}-2\mu \varphi_{21}^{'}+g_1-h_0^{'}}<<1,
\end{equation}
then there exist $t_0>0$ and $\bar{x}\in \Gamma$ such that boundary layer separation represented by the solution of (\ref{eql1}) occur  at $(t_0,\bar{x})$,
where $\varphi_{11}, \varphi_{13}, \varphi_{21}, g_1$ and $h_0$ as (\ref{equ3})-(\ref{equ6}).
\end{theorem}
\begin{proof}

$u$ has the Taylor expansion at $t=0$,
\begin{equation}\label{equ10}
u=\tilde{\varphi}+t\omega+o(t^2),
\end{equation}
where $\omega=(\omega_1,\omega_2)$ satisfies
\begin{equation}\label{equ11}
\omega|_{\partial \Omega}=0, \ \ \ divw=0.
\end{equation}

Thus,
\begin{equation}\label{equ12}
\omega=(x_2\omega_1+o(x_2), x^2_2\omega_2+o(x_2^2)).
\end{equation}

Let $p=p_0+tp_1+o(t)$ and $p_0=p_{01}(x_1)+x_2p_{02}(x_1)+o(x_2)$.

Substituting $u$ in (\ref{eql1}), we get
\begin{equation*}
x_2 \omega_1=2\mu \varphi_{12}+\mu x_2 \big(\varphi_{11}^{''}+6 \varphi_{13}\big)
\end{equation*}
\begin{equation}\label{equ13}
-\frac{1}{\rho}(p_{01}^{'}+x_2p_{02}^{'})+g_0+x_2 g_1+o(x_2),
\end{equation}
\begin{equation}\label{equ13}
0=2\mu\varphi_{21}-\frac{1}{\rho}p_{02}+h_0+o(1)
\end{equation}

Compare the coefficient on both sides to get
\begin{equation}\label{equ14}
p_{02}=\rho (h_0+2\mu \varphi_{21}).
\end{equation}
\begin{equation}\label{equ15}
\omega_1=\mu \big(\varphi_{11}^{''}+6 \varphi_{13}\big)-\frac{1}{\rho}\frac{d p_{02}}{d x_1}++g_1.
\end{equation}

Combine (\ref{equ14}) and (\ref{equ15}) to get
\begin{equation}\label{equ15}
\omega_1=\mu \big(\varphi_{11}^{''}+6 \varphi_{13}\big)-h_0^{'}-2\mu \varphi_{21}^{'}+g_1.
\end{equation}

With (\ref{equ10}) and (\ref{equ12}), we know
$$
\frac{\partial u_\tau}{\partial n}|_{\Gamma}=\frac{\partial u_1}{\partial x_2}|_{\Gamma}
$$
$$
=\frac{\partial \tilde{\varphi}_1}{\partial x_2}+t \omega_1+o(t)|_{\Gamma}
$$
$$
=\varphi_{11}+t (\mu \varphi_{11}^{''}+6 \mu\varphi_{13}-h_0^{'}-2\mu \varphi_{21}^{'}+g_1)+o(t)|_{\Gamma}
$$

Clearly, if $t_0=\min_{x_1 \in \Gamma}\frac{-\varphi_{11}}{\mu \varphi_{11}^{''}+6 \mu\varphi_{13}-2\mu \varphi_{21}^{'}+g_1-h_0^{'}}<<1$, we obtain
$$
\frac{\partial u_\tau}{\partial n}|_{\Gamma}=0.
$$

Then, there exists one point $\bar{x} \in \Gamma$ which makes
$$
\frac{-\varphi_{11}}{\mu \varphi_{11}^{''}+6 \mu\varphi_{13}-2\mu \varphi_{21}^{'}+g_1-h_0^{'}}
$$
obtain minmum.

Hence, $\bar{x}$ is an isolated singularity of $u(\cdot, t_0)$ on $\partial\Omega$ , and
$$
\frac{\partial u_\tau(\bar{x}, t_0)}{\partial n}|_{\Gamma}=0.
$$

That is, there exists $\bar{x}$ which satisfies (\ref{equ8}).

Clearly,
$$
\frac{\partial u_\tau}{\partial n}|_{\Gamma}=\frac{\partial u_1}{\partial x_2}|_{\Gamma}
$$
$$
=\frac{\partial \tilde{\varphi}_1}{\partial x_2}+t \omega_1+o(t)|_{\Gamma}
$$
$$
=\varphi_{11}+t (\mu \varphi_{11}^{''}+6 \mu\varphi_{13}-h_0^{'}-2\mu \varphi_{21}^{'}+g_1)+o(t)|_{\Gamma}\neq 0,
$$
if $t<t_0$.

That is, there exists $t<t_0$ which satisfies (\ref{equ7}). The proof is complete.
\end{proof}

\section{Physical interpretation}

 In the following, we give the physical interpretation related to
Theorem 1.

 1. In this paper, we analyze the solutions of the Navier-Stokes equations governing incompressible boundary layer flows and get the condition (\ref{6}) for boundary layer separation. $\varphi_{11}, \varphi_{21}, g_1$, $\varphi_{13}$  and $h_0$ in the condition (\ref{6}) are observations. $t_0=\min_{\Gamma} \frac{-\varphi_{11}}{\mu \varphi_{11}^{''}+6\mu \varphi_{13}-2\mu \varphi_{21}^{'}+g_1-h_0^{'}}$ is the time at which boundary layer separation occurs. It is $\bar{x}$ that  makes $\frac{-\varphi_{11}}{\mu \varphi_{11}^{''}+6\mu \varphi_{13}-2\mu \varphi_{21}^{'}+g_1-h_0^{'}}$ obtain minimum is the position in which boundary layer separation occurs.

2. The time $t_0$ in Theorem 1 is nondimensional. We should converse $t_0$ to the dimensional time  $T_0$ in practical application.

3. If $-\varphi_{11}$ and $\mu \varphi_{11}^{''}+6\mu \varphi_{13}-2\mu \varphi_{21}^{'}+g_1-h_0^{'}$ are linearly dependent,
there is more than one $\bar{x}$ such that make $\frac{-\varphi_{11}}{\mu \varphi_{11}^{''}
+6\mu \varphi_{13}-2\mu \varphi_{21}^{'}+g_1-h_0^{'}}$ obtain minimum.
That is, $\bar{x}$ is not an isolated singularity. However. there is not a possibility that $\mu \varphi_{11}^{''}+6\mu \varphi_{13}-2\mu \varphi_{21}^{'}+g_1-h_0^{'}$ are linearly dependent in natural phenomena. Here, we don't have special emphasis on it.

4. If external force is zero, (\ref{6}) becomes $0< \min_{\Gamma} \frac{-\varphi_{11}}{\mu \varphi_{11}^{''}+6\mu \varphi_{13}-2\mu \varphi_{21}^{'}}<<1$.
Clearly, if the observational velocity $\varphi$ satisfies $0< \min_{\Gamma} \frac{-\varphi_{11}}{\mu \varphi_{11}^{''}+6\mu \varphi_{13}-2\mu \varphi_{21}^{'}}<<1$, then boundary layer separation will occur.

\section{Application}

\begin{example}
  Suppose the airflows near the wing as shown  in figure 2(a). $(\varphi_1,\varphi_2)$ is the velocity of boundary layer at the bottom of the wing at some time. Let $o$ be the origin of coordinates in Figure 2(a). $(\varphi_1,\varphi_2)$ has the Taylor expansion at $o=0$,
\begin{eqnarray}\label{18s}
 \varphi_1=(v_0-2\alpha x_1-\beta x_1^2+o(x_1^2))x_2+o(x_2),
\end{eqnarray}
\begin{eqnarray}
\varphi_2=(\alpha+\beta x_1+o(x_1))x_2^2+o(x_2^2),
\end{eqnarray}
where  $x_1$ is horizontal axes, $x_2$ is vertical axes , $v_0$ is the speed of the wing and $\beta$ is the decay rate of the velocity representing the
frictional force. Physically, $\beta$ is has a relation with the velocity of airplane as $\beta=\gamma v_0^k(k>1)$.

$(\varphi_1(0,x_2),\varphi_2(0,x_2))=(\varphi_1(0,x_2),0)$ means $\alpha=0$. For the boundary layer at the bottom of the wing, we know that
$$\varphi_{11}=v_0-\beta x_1^2+o(x_1^2), \varphi_{11}^{''}=-2\beta, \varphi_{21}^{'}=\beta,$$

$$\varphi_{13}=g_1=h_0^{'}=0.$$

Then $$\frac{-\varphi_{11}}{\mu \varphi_{11}^{''}+6\mu \varphi_{13}-2\mu \varphi_{21}^{'}+g_1-h_0^{'}}=\frac{-v_0+\beta x_1^2+o(x_1^2)}{-4\mu\beta}=\frac{v_0}{4\mu\beta}-\frac{x_1^2}{4\mu}+o(x_1^2).$$

Obviously, we get
 $$0<\frac{v_0}{4\mu\beta}-\frac{x_1^2}{4\mu}+o(x_1^2)<<1.$$
 For any $s<\sqrt{\frac{v_0}{\beta}}=\sqrt{\frac{1}{\gamma v_0^{k-1}}}$ and $x_1\in(0,s]$,
 \begin{eqnarray}\label{17s}
 0<\frac{v_0}{4\mu\beta}=\frac{1}{4\mu\gamma v_0^{k-1}}<<1,
 \end{eqnarray}
for the velocity $v_0$ of airplane is large.

 Hence, we obtain
 \begin{eqnarray*}
 \min_{0<x_1\le s}\frac{-\varphi_{11}}{\mu \varphi_{11}^{''}+6\mu \varphi_{13}-2\mu \varphi_{21}^{'}+g_1-h_0^{'}}
  \end{eqnarray*}
 \begin{eqnarray}\label{cc}
 \approx\min_{0<x_1\le s }\frac{1}{4\mu\gamma v_0^{k-1}}-\frac{x_1^2}{4\mu}
 \end{eqnarray}

$(\ref{cc})$ means that a vortex will appear in the boundary layer at time
\begin{eqnarray}\label{25}
T_s=\frac{1}{4\mu\gamma v_0^{k-1}}-\frac{s^2}{4\mu}.
\end{eqnarray}
The vortex is at $x_1=s$.
$(\ref{25})$ means that $T_{s_1}<T_{s_2}$ for $s_2<s_1$.

\begin{figure}[h]
 \begin{center}
 \includegraphics[width=9cm]{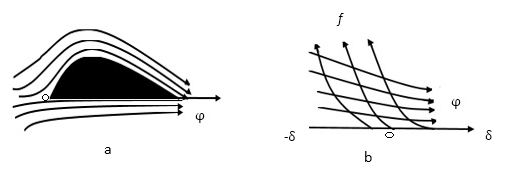} \\
 \end{center}
   \caption{}
 \end{figure}

With $(\ref{cc})$ and $(\ref{25})$, we can get the following conclusions.

1. There will be many vortices near $x_1=0$ in a short time, if the velocity $v_0$ of airplane is very large.

2. The vortices occur in different position near $x_1=0$ at different time.

3. The vortex occur in $x_1=s_1$ more early than in $x_1=s_2$, if $s_1>s_2$.
\end{example}

\begin{example}
Suppose current drift as shown in Figure 2(b) at some time. $f=(f_1,f_2)$ is the strength of sea wind. $\varphi$ is the velocity of the current drift. $(-\delta,\delta)$ is one part of coast.

 Let $o$ be the origin of coordinates in Figure 2(b). Here, we take
  \begin{eqnarray}
 \label{17}f_1=-\frac{1}{2\delta+x_1}\\
 f_2=-\frac{1}{2}x_1^2+2\delta x_1+5\delta^2\\
 \varphi_1=\frac{\theta}{2\delta+x_1}x_2\\
 \varphi_2=\frac{\theta}{2(2\delta+x_1)^2}x_2^2\label{21},
 \end{eqnarray}
where  $x_1$ is horizontal axes, $x_2$ is vertical axes.

Substituting $(\ref{17})-(\ref{21})$ in $(\ref{6})$, we get
\begin{eqnarray*}
0<\frac{-\varphi_{11}}{\mu \varphi_{11}^{''}+6\mu \varphi_{13}-2\mu \varphi_{21}^{'}+g_1-h_0^{'}}
\end{eqnarray*}
\begin{eqnarray*}
=\frac{\theta(2\delta+x_1)^2}{(2\delta+x_1)^3( -x_1+
2\delta)-4\mu\theta}
\end{eqnarray*}
\begin{eqnarray}\label{22}
<\frac{9\theta}{\delta^2}<<1(\delta>>1)
\end{eqnarray}

With $(\ref{22})$, we obtain
\begin{eqnarray*}
\min_{-\delta<x_1<\delta}\frac{\theta(2\delta+x_1)^2}{(2\delta+x_1)^3( -x_1+
2\delta)-4\mu\theta}
\end{eqnarray*}
\begin{eqnarray}\label{23}
\approx\min_{-\delta<x_1<\delta}\frac{\theta}{(2\delta+x_1)( -x_1+
2\delta)}=\frac{\theta}{4\delta^2}
\end{eqnarray}

$(\ref{23})$ means that the current drift as shown in Figure 2(b) will separate a vortex at time $T=\frac{\theta}{4\delta^2}$, and the vortex is at $x_1=0$.
\end{example}


\begin {thebibliography}{90}


\bibitem{Blanchonette} P. Blanchonette and M. A. Page. Boundary-Layer Separation in the Two-Layer Flow Past a Cylinder in a Rotating Frame. Theoret Comput Fluid Dynamics, (1998) 11: 95--108

\bibitem{Chorin} A. Chorin, J. Marsden. A mathematical Introduction to Fluid Mechanics. Springer-Verlag, 1997

\bibitem{W. E} W. E, B. Engquist, Blow up of solutions of the unsteady Prandtl¡¯s equation, Commun Pure Appl Math, (1997)50: 1287--1293

\bibitem{Ghil1} M. Ghil, T. Ma, S. Wang, Structural bifurcation of 2-D incompressible flows, Indiana Univ Math J, (2001)50: 159--180

\bibitem{Ghil} M. Ghil, J. Liu, and C. Wang and S. Wang, Boundary-layer separation and adverse pressure gradient for 2-D viscous incompressible flow, Physica D, (2004) 197: 1--2, 149--173

\bibitem{Ma7} M. Ghil and T. Ma and S. Wang, Structural Bifurcation of 2-D Incompressible Flows with Dirichlet Boundary Conditions: Applications to Boundary-Layer Separation, SIAM J Applied Math, (2005) 65(5): 1576--1596

\bibitem{Larin} O. B. Larin and V. A. Levin. Boundary Layer Separation in a Laminar Supersonic Flow with Energy Supply Source. Technical Physics Letters, (2008) 34(3): 181--183.


\bibitem{Larin1}  O. B. Larin and V. A. Levin. Effect Of Energy Supply To A Gas On Laminar Boundary Layer Separation. J Appl  Mech  and Tech  Phys, (2010) 51(1): 11--15

 \bibitem{Liu}   J. Liu, Z. Xin, Boundary layer behavior in the fluid-dynamic limit for a nonlinear model Boltzmann equation, Arch Rat Mech Anal, (1996)135: 61--105

\bibitem{Ma5} T. Ma and S. Wang, Rigorous Characterization of Boundary Layer Separations, Proc. of the Second MIT Conference on Computational Fluid and Solid Mechanics, Cambridge, MA, 2003

\bibitem{Ma2} T. Ma and S. Wang, Interior Structural Bifurcation and Separation of $2-D$ Incompressible Flows, J Math Phy, (2004)45(5): 1762--1776

\bibitem{Ma3} T. Ma and S. Wang, Asymptotic Structure for Solutions of the Navier-Stokes Equations, Disc Cont Dyna Syst, (2004) 11(1) 189--204

\bibitem{Ma4} T. Ma and S. Wang, Boundary Layer Separation and Structural Bifurcation for 2-D Incompressible Fluid Flows, Disc Cont Dyna Syst, (2004)10(1--2): 459--472

\bibitem{Ma6} T. Ma and S. Wang, Geometric Theory of Incompressible Flows with Applications to Fluid Dynamics, AMS Mathematical Surveys and Monographs Series, vol. 119, 2005, 234 pp


\bibitem{Ma8} T. Ma and S. Wang, Boundary Layer and Interior Separations in the Taylor--Couette--Poiseuille Flow, J Math Phys, (2009)50(3): 1--29


\bibitem{Oleinik1} O. Oleinik, On the mathematical theory of boundary layer for unsteady flow of incompressible fluid, J Appl Math Mech, (1966)30:951--974.

\bibitem{Oleinik} O. Oleinik and V. N. Samokhin, Mathematical Models in Boundary Layer Theory, Chapman and Hall/CRC, Boca Raton, FL, 1999.


\bibitem{Prandtl} L. Prandtl. On the motion of fluids with very little friction, in Verhandlungen des dritten internationalen Mathematiker-Konggresses, Heidelberg, 1904, Leipeizig, 1905, 484--491

\bibitem{Schlichting} H. Schlichting, Boundary Layer Theory, eighth ed., Springer, Berlin-Heidelberg, 2000.

\bibitem{Smith} F.T. Smith and S.N. Brown. Boundary-Layer Separation. Proceedings of the IUTAM Symposium London, 1986, Springer-Verlag, 1987


\end{thebibliography}
\end{document}